%% file: iterEqF_arxiv.tex
\documentclass{article}
\usepackage{arxiv}

\input{preamble}

\newtheorem{theorem}{Theorem}[section]

\usepackage[backend=bibtex,bibstyle=ieee,citestyle=numeric-comp,doi=false,isbn=false]{biblatex} 
\title{The Iterative Equivariant Filter}
\headertitle{The Iterative Equivariant Filter}

\author{
\href{https://orcid.org/0000-0003-4391-7014}{\includegraphics[scale=0.06]{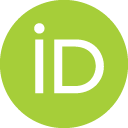}\hspace{1mm}
Pieter van Goor}
\\
    School of Aerospace, Mechanical, and Mechatronic Engineering \\
    The University of Sydney \\
    NSW 2006, Australia \\
    \texttt{pieter.vangoor@sydney.edu.au} \\
	\And	\href{https://orcid.org/0000-0002-1987-9268}{\includegraphics[scale=0.06]{orcid.png}\hspace{1mm}
    James Richard Forbes}
\\
    Dept.~of Mechanical Engineering \\
	McGill University \\
    Montreal, QC, Canada \\
	\texttt{james.richard.forbes@mcgill.ca} \\
}

\newcommand{\mmap}{\Theta}

\newcommand{\publicationdetails}
{Under review}

\newcommand{\publicationversion}
{Preprint}

\begin{document}




\maketitle

\begin{abstract}
This paper presents the iterative equivariant filter (IterEqF). 
The iterative extended Kalman filter (IterEKF) replaces the standard EKF correction step with an iterative correction step that is the Gauss-Newton solution to a nonlinear weighted least squares problem.
The standard equivariant filter (EqF) exploits and respects the underlying symmetry of state estimation problems posed on homogeneous spaces and Lie groups. 
The motivation behind the IterEqF is to combine the features of both the IterEKF and the EqF, thus leading to a high-performance state estimation solution that is well suited to navigation problems. 
This paper derives the iteration procedure for the IterEqF update step and shows how the intrinsic nonlinearity of the approach naturally results in a `reset' of the filter's covariance into new coordinates.
Monte-Carlo simulations of range-based localisation for a mobile robot demonstrate the improvement in performance relative to a standard EqF, especially during the transient convergence.
\end{abstract}

\section{Introduction}

Mobile robotic systems such as uncrewed ground, underwater, and aerial vehicles, move in three-dimensional space as they execute tasks including material handling, inspection of infrastructure, mapping, and other missions. 
To execute these tasks reliably, the robot must be able to reliably estimate states of interest, such as position, velocity, and attitude.

The state of a robotic system is often naturally represented by a Lie group element.
For example, $\SO(3)$ elements represent attitude, $\SE(3)$ elements represent position and attitude, and $\SE_2(3)$ elements represent position, velocity, and attitude.
Recently, there has been a focused effort on deterministic observer design for mobile robotic systems, where the observer state is an element of $\SO(3)$ \cite{paper_Kinsey_Whitcomb_TRO_2007,paper_mahony_2008}, $\SE(3)$ \cite{paper_Hua_etal_CDC_ECC_2011,paper_Hua_etal_CDC_2015,paper_Zlotnik_Forbes_TAC_2019}, $\SE_2(3)$ \cite{paper_Berkane_arXiv_2026}, $\SE_5(3)$ \cite{paper_Benahmed_etal_CDC_2025}, or $\SE_{3+n}(3)$ \cite{paper_Boughellaba_etal_arXiv_2026}.
Observers solving the simultaneous localization and mapping (SLAM) problem have also been considered \cite{paper_Mahony_Hamel_CDC_2017,paper_Zlotnik_Forbes_TAC_2018}, again exploiting the geometry of the SLAM problem. 
In parallel, the application of Lie group tools within the extended Kalman filter (EKF) has also been an active area of research.
The aerospace community has long promoted quaternion-based variants of the EKF, such as the multiplicative EKF (MEKF) \cite{paper_Lefferts_Markley_Shuster_1982,paper_Markley_JGCD_2003} and norm-constrained EKF \cite{paper_Zanetti_Majji_Bishop_Morari_2009,paper_Forbes_de_Ruiter_norm_cont_2013}.
Within the robotics community, MEKFs that work directly with Lie group representations of the state, such as $\SE(3)$ \cite[\S~10]{paper_Sola_eta_al_micro}, have become common due to their well-documented performance advantages over EKFs based on Euclidean coordinates.

The invariant EKF (InEKF) is a particular type of MEKF that, when interpreted as an observer, possesses desirable convergence properties \cite{Barrau2017,Barrau2018}.
The advantage of the InEKF is that, when interpreted as an observer, the Jacobians associated with linearization are independent of the state estimate, and the resulting domain of convergence is independent of the system state.
However, these results rely on a group affine process model and right- or left-invariant measurements.
The equivariant filter (EqF) is a different extension of the MEKF applicable to systems on homogeneous spaces with a Lie group state symmetry \cite{van_Goor_etal_CDC_2020,moahny_etal_ARC_2022}.
The keys to applying the EqF are defining a symmetry group, a transitive group action, and a lift.
These enable the filter's mean to evolve on the Lie group while its covariance equations can be posed on the original system state space.
In the context of inertial navigation system (INS) symmetries, the EqF framework can recover the classic MEKF, InEKF, ``imperfect" InEKF, two-frame group IEKF, and others, depending on the choice of symmetry group \cite{paper_Fornasier_etal_Automatica_2025}.
A well-chosen symmetry can significantly enhance EqF performance compared to classic and standard EKFs \cite{paper_Fornasier_etal_Automatica_2025}.

EKF variants are generally executed in a predict-correct sequence, where the process model is used to predict the state, and the measurement model is used to correct the state.
The iterative EKF (IterEKF) is yet another variant of the EKF that maintains the standard prediction step of the EKF, but modifies the correction step by iteratively updating the corrected state until a convergence criterion is met, leading to enhanced filter performance and robustness \cite{paper_Bell_Cathey_1993,paper_Bourmaud_iteritive_LG_IKF,paper_Liu_Chen_Zhang_CDC_2023}.
The IterEKF correction step can be derived by posing and solving a nonlinear weighted least-squares problem using the Gauss-Newton algorithm  \cite{paper_Bell_Cathey_1993}.
Iteration of the correction step has been considered in the context of a matrix Lie group EKF in \cite{paper_Bourmaud_iteritive_LG_IKF,paper_Liu_Chen_Zhang_CDC_2023} and an InEKF in \cite{Goffin_etal_2026}.

This paper presents the iterative EqF (IterEqF). 
The motivation is to incorporate iteration into the EqF framework, thus combining the advantages of symmetry-based EKFs with the improved transient performance and robustness of the iterative update step. 
The correction step of the IterEqF is the Gauss-Newton solution of a nonlinear weighted least-squares problem weighing information from the prediction step and the measurement, lifted from the homogeneous state space to the chosen symmetry group. 
Special care is taken to define the iteration on the Lie group while accounting for degenerate directions associated with the lifting, and it is shown that the ``reset step" \cite{paper_Markley_JGCD_2003,paper_Mueller_etal_JGCD_2017,paper_Ge_etal_CDC_2022,Paper_Ge_van_Goor_Mahony_LCSS_2024,ge_etal_CEP_2026} is a natural consequence of the IterEqF correction.
This is an important extension of prior work \cite{paper_Bourmaud_iteritive_LG_IKF,paper_Liu_Chen_Zhang_CDC_2023} that showed the reset step appears for iterative EKFs posed directly on Lie groups, and contrasts with \cite{Lu_etal_2025} where the reset step is applied after iteration as an ad-hoc addition similarly to the standard EqF \cite{paper_Ge_etal_CDC_2022}.
Monte-Carlo simulations of a range-based localization problem demonstrate that the IterEqF has superior transient performance, in both the error and covariance, relative to the standard EqF.

The remainder of this paper is as follows. Section~\ref{sec:prelim} provides the requisite mathematical foundation relevant to equivariance. The standard EqF is reviewed in Section~\ref{sec:EqF_review}. The main contribution, the IterEqF, is derived in Section~\ref{sec:iterEqf}. Simulation results are presented in Section~\ref{sec:example}. Conclusions are drawn in Section~\ref{sec:conclusion}

\section{Preliminaries}
\label{sec:prelim}

For an introduction to smooth manifolds and Lie groups, the authors recommend \cite{2012_lee_IntroductionSmoothManifolds}.

Consider a smooth ($C^\infty$) manifold $\calM$.
The tangent space at $\xi \in \calM$ is denoted $\T_\xi \calM$, and the tangent bundle is denoted $\T \calM$.
For a smooth map $h : \calM \to \calN$, where $\calN$ is another manifold, the differential of $h$ at a point $\xi \in \calM$ is written
\begin{align*}
    \diff h(\xi) &: \T_\xi\calM \to \T_{h(\xi)}\calN.
\end{align*}
The base point may also be omitted, in which case we regard $\diff h : \T \calM \to \T\calN$ as a map between the tangent bundles.
For a map $h : \calM_1 \times \calM_2 \to \calN$, we write the partial maps
\begin{gather*}
    h_\xi : \calM_2 \to \calN, \qquad h^\zeta : \calM_1 \to \calN, \\
    h_\xi(\zeta) = h(\xi,\zeta) = h^\zeta(\xi),
\end{gather*}
for all $\xi \in \calM_1$ and $\zeta \in \calM_2$.

A Lie group $\grpG$ is a smooth manifold equipped with a smooth group structure.
We write the identity as $I \in \grpG$.
The product of $X,Y \in \grpG$ is written as $XY$, and the inverse of $X$ is written as $X^{-1}$.
The Lie algebra of $\grpG$ is written as $\gothg$ and may be identified with the tangent space of $\grpG$ at the identity.
The exponential $\exp : \gothg \to \grpG$ is a smooth map from the Lie algebra to the Lie group, and it is locally invertible.
Over the neighbourhood $\calU \subset \grpG$ of the identity where $\exp$ is invertible, its inverse is written $\log : \calU \to \gothg$.
The group and algebra adjoint actions are written as $\Ad : \grpG \times \gothg \to \gothg$ and $\ad : \gothg \times \gothg \to \gothg$, respectively.
The wedge and vee maps are linear isomorphisms $\cdot^\wedge : \R^{\dim \gothg} \to \gothg$ and $\cdot^\vee : \gothg \to \R^{\dim \gothg}$ obtained by selecting a basis for $\gothg$.

A (right\footnote{We consider only right group actions in this paper.}) group action is a smooth map $\phi : \grpG \times \calM \to \calM$, where $\grpG$ is a Lie group and $\calM$ a smooth manifold, satisfying
\begin{align}
    \phi(I, \xi) &= \xi, \\
    \phi(X_2, \phi(X_1, \xi)) &= \phi(X_1 X_2, \xi),
\end{align}
for all $X_1, X_2 \in \grpG$ and $\xi \in \calM$.
If the partial map $\phi^\xi : \grpG \to \calM$ is injective for all $\xi \in \calM$, then we say that $\phi$ is \emph{free}.
If the partial map $\phi^\xi : \grpG \to \calM$ is surjective for some $\xi \in \calM$, then it is necessarily surjective for all $\xi \in \calM$, and we say that $\phi$ is \emph{transitive}, and we say that $\calM$ is a \emph{homogeneous space}.

The space of $m\times m$ symmetric positive definite matrices is written as $\Sym_+(m) \subset \R^{m\times m}$.
The normal distribution with mean $\mu \in \R^m$ and covariance $\Sigma \in \Sym_+(m)$ is written as $N(\mu, \Sigma)$.

\section{Equivariant Filtering}
\label{sec:EqF_review}

This section reviews the Equivariant Filter design methodology for systems with continuous dynamics and discrete measurements.
Consider the following system evolving on a state space manifold $\calM$ of dimension $m$, with inputs in $\R^l$ and outputs in $\R^n$,
\begin{subequations}
\begin{align}
    \dot{\xi} &= f_u(\xi), \label{eq:system_dynamics} \\
    y &= h(\xi) \label{eq:output_y}.
\end{align}
\end{subequations}
Here, $\xi \in \calM$ is the state, $u \in \R^l$ is the input, $f : \R^l \to \mathfrak{X}(\calM)$ is the system function, $y$ is the output, and $h : \calM \to \R^n$ is the measurement function.
The system dynamics \eqref{eq:system_dynamics} are considered to flow continuously for all time, while the measurements $y$ are considered to arrive only at a discrete sequence of times labelled $t_1, t_2, \ldots \in \R$ satisfying $t_{k+1} > t_k$ for all $k \in \N$.

Equivariant filter design requires a choice of a transitive right group action $\phi: \grpG \times \calM \to \calM$.
This selection then defines the relationship between the group $\grpG$ and the underlying state space $\calM$.
The next step is to choose a map $\Lambda : \calM \times \R^l \to \mathfrak{g}$, satisfying
\begin{align}\label{eq:lift_condition}
    \diff \phi^\xi (I) \Lambda(\xi, u) = f_u(\xi),
\end{align}
for all $\xi \in \calM$ and all $u \in \R^l$.
This is termed the \emph{lift} and defines how trajectories of the system are replicated on the group.
The lift condition \eqref{eq:lift_condition} constrains $\Lambda(\xi,u)$ except for the values it takes along the Lie algebra of the stabiliser of $\phi$ at $\xi$, meaning that any two lifts $\Lambda^1$ and $\Lambda^2$ must satisfy
\begin{align*}
    \Lambda^1(\xi, u) - \Lambda^2(\xi,u) \in \stab_\phi(\xi),
\end{align*}
for all $\xi \in \calM$ and $u \in \R^l$.
In particular, if $\phi$ is a free group action, then $\stab_\phi(\xi) = \{0\}$ and the lift $\Lambda$ is uniquely determined by \eqref{eq:lift_condition}.

Choose a fixed \emph{origin} $\xi_0 \in \calM$ and define the observer state $\hat{X} \in \grpG$ to have dynamics
\begin{align}
    \dot{\hat{X}} &= \hat{X} \Lambda(\phi(\hat{X}, \xi_0), u_m), \\
    \hat{X}(t_k^+) &= \exp(\Delta_k) \hat{X}(t_k^-),
\end{align}
where $u_m \in \R^l$ is the measured input, $\Delta_k \in \mathfrak{g}$ is a correction term that remains to be designed, and $t_k^-$ and $t_k^+$ represent the times instantaneously before and after receiving a measurement at time $t_k$.
Then the state estimate is given by $\hat{\xi} = \phi(\hat{X}, \xi_0)$, and satisfies
\begin{align*}
    \dot{\hat{\xi}} &= f_u(\hat{\xi}), \\
    \hat{\xi}(t_k^+) &= \phi^{\hat{\xi}}(\exp(\Ad_{\hat{X}^{-1}} \Delta_k)).
\end{align*}
If $\Delta_k = 0$, then the dynamics of $\hat{\xi}$ are exactly those of the original system \eqref{eq:system_dynamics}.

The \emph{equivariant error} is defined as
\begin{align}
    e = \phi(\hat{X}^{-1}, \xi) \in \calM.
\end{align}
Choose $\vartheta : U_0 \subseteq \calM \to \R^m$ to be a chart of $\calM$ in a neighbourhood $U_0$ of $\xi_0$, then the local error coordinates are given by $\varepsilon = \vartheta(e)$.
If there is a map $\mmap : \R^{m} \to \gothg$ such that $\vartheta^{-1}(\varepsilon) = \phi(\exp(\mmap \varepsilon), \xi_0)$ and $\gothg = \im \mmap \oplus \stab_\phi(\xi_0)$, then we say that $\vartheta$ is a \emph{normal} chart.
In the case where $\calM$ is isomorphic to $\grpG$ (the group action is free), then the normal charts are exactly the charts given through the Lie group logarithm, differing only by a choice of Lie algebra basis.
We will assume that $\vartheta$ is a normal chart.

The Equivariant filter directly estimates the distribution of the local error coordinates $\varepsilon$ as a zero-mean Gaussian with covariance $\Sigma \in \Sym_+(m)$.
The continuous dynamics of $e$ are directly seen to be
\begin{align*}
    \dot{e}
    &= \ddt \phi(\hat{X}^{-1}, \xi) \\
    &= \diff \phi_{\hat{X}^{-1}}(\xi)[ f_u(\xi) ] 
    - \diff \phi^{\xi} (\hat{X}^{-1})[\Lambda(\phi(\hat{X}, \xi_0), u_m)\hat{X}^{-1}] \\
    &= \diff \phi_{\hat{X}^{-1}}(\xi) \diff \phi^\xi (I) \Lambda(\xi, u)
    \\ &\hspace{1cm}
    - \diff \phi^{\phi(\hat{X}^{-1}, \xi)} (I)[\hat{X}\Lambda(\phi(\hat{X}, \xi_0), u_m)\hat{X}^{-1}] \\
    &= \diff \phi^{\phi(\hat{X}^{-1}, \xi)} (I) \Ad_{\hat{X}} \Lambda(\xi, u)
    - \diff \phi^{e} (I) \Ad_{\hat{X}}\Lambda(\hat{\xi}, u_m) \\
    &= \diff \phi^{e} (I) \Ad_{\hat{X}} [\Lambda(\phi_{\hat{X}}(e), u) - \Lambda(\hat{\xi}, u_m)].
\end{align*}
The dynamics of $\varepsilon$ are then simply the dynamics of $e$ expressed in local coordinates and linearised about $\varepsilon \approx 0$ and $u \approx u_m$.
Specifically,
\begin{subequations}
\begin{align}
    \dot{\varepsilon}
    &= A_t \varepsilon + B_t (u-u_m), \label{eq:error_dynamics_lineariased} \\
    A_t &= \diff \vartheta \diff \phi^{\xi_0}(I) \Ad_{\hat{X}} \diff \Lambda^{u_m}(\hat{\xi}) \diff \phi_{\hat{X}} \diff \vartheta^{-1}, \label{eq:A_matrix} \\
    B_t &= \diff \vartheta \diff \phi^{\xi_0}(I) \Ad_{\hat{X}} \diff \Lambda_{\hat{\xi}}(u_m). \label{eq:B_matrix}
\end{align}
\end{subequations}
Based on this linearisation, the covariance $\Sigma$ of the EqF is assigned the dynamics
\begin{align*}
    \dot{\Sigma} = A_t \Sigma + \Sigma A_t^\top + B_t Q_t B_t^\top,
\end{align*}
where $Q_t \in \Sym_+(l)$ represents the noise associated with the measurement of $u_m$, typically $u-u_m \sim N(0, Q_t)$.

\section{Iterative Updates for Equivariant Filters}
\label{sec:iterEqf}

The measurement update for discrete time measurements in an EqF was first described in \cite{paper_Ge_etal_CDC_2022}, where a single step is computed to update the filter's estimated state and covariance.
This paper instead takes an iterative approach to the update step, framing it as an optimization procedure that aims to find the optimal state estimate and covariance, given all the information available at the time.
Upon receiving a new measurement, the EqF possesses two independent pieces of information about the state $\xi$.
First, according to the filter's own state and covariance, 
\begin{align*}
    \vartheta(\phi(\check{X}^{-1}, \xi)) \sim N(0, \check{\Sigma}),
\end{align*}
where $\check{X} = \hat{X}(t_k^-) \in \grpG$ and $\check{\Sigma} = \Sigma(t_k^-) \in \Sym_+(m)$ capture the filter's information state immediately \emph{prior} to the receipt of a new measurement.
Second, according to the new measurement,
\begin{align*}
    y - h(\xi) \sim N(0, R),
\end{align*}
where $R \in \Sym_+(n)$ is the covariance of the noise associated with the measurement process.
These two items can be combined into a cost function $S : \calM \to \R^+$ for the unknown state $\xi$, defined by
\begin{align}
    S(\xi)
    &:= \frac{1}{2} r(\xi)^\top \mathrm{diag}(\check{\Sigma}^{-1}, R^{-1}) r(\xi), \\
    r(\xi)
    &:= \begin{pmatrix}
        \vartheta(\phi(\check{X}^{-1}, \xi)) \\
        y - h(\xi)
    \end{pmatrix},
\end{align}
where $r : \calM \to \R^{m+n}$ is termed the \emph{residual}.
Unlike a classical, extended, or invariant Kalman filter, the EqF state belongs to the group $\grpG$ rather than the state space $\calM$.
This means that, to optimise over a group element, the cost function and residual definitions have to be lifted to $\tilde{S} : \grpG \to \R^+$ and $\tilde{r} : \grpG \to \R^{m+n}$, defined by
\begin{align}
    \tilde{S}(X) &:=
    S(\phi(X, \xi_0)) \notag\\
    &= \frac{1}{2} \tilde{r}(X)^\top \mathrm{diag}(\check{\Sigma}^{-1}, R^{-1}) \tilde{r}(X), \label{eq:lifted_cost}\\
    \tilde{r}(X) &:= r(\phi(X, \xi_0)) \notag\\
    &= \begin{pmatrix}
        \vartheta(\phi(X \check{X}^{-1}, \xi_0)) \\
        y - h(\phi(X,\xi_0))
    \end{pmatrix}.
\end{align}
The following theorem details how to compute each iteration of a Gauss-Newton algorithm to estimate $\hat{X}$ that optimises $\tilde{S}$.
Upon receiving a measurement, the IterEqF computes the Gauss-Newton iterations described below until convergence or until a maximum number of iterations is reached.

\begin{theorem}\label{thm:iterations}
    Let $\hat{X}_0 = \check{X} \in \grpG$ and $\hat{\Sigma}_0 = \check{\Sigma} \in \Sym_+(m)$ be initial estimates of the optimum and inverse Hessian of the cost function $\tilde{S}$ as defined in \eqref{eq:lifted_cost}.
    Then the Gauss-Newton iterations are defined for each $j=0,1,...$ by
    \begin{subequations}    
    \begin{align}
        \varepsilon_j &= \vartheta(\phi(\hat{X}_j \check{X}^{-1}, \xi_0)) \\
        \label{eq:check_Sigma_dfn}
        \check{\Sigma}_j &= J_j^{-1} \check{\Sigma} J_j^{-\top} \\
        K_j &= \check{\Sigma}_j C_j^\top \left( C_j \check{\Sigma}_j C_j^\top + R \right)^{-1} \\
        \delta_j &= K_j (y - h(\phi(\hat{X}_j, \xi_0)) + C_j J_j^{-1} \varepsilon_j) 
        - J_j^{-1} \varepsilon_j \\
        \hat{X}_{j+1} &= \exp(\alpha_j \mmap \delta_j) \hat{X}_j \\
        \hat{\Sigma}_{j+1} &= (I_m - K_j C_j) \check{\Sigma}_j,
    \end{align}
    \end{subequations}
    where $J_j \in \R^{m\times m}$ and $C_j \in \R^{n \times m}$ are defined by
    \begin{align}
        J_j &= \diff \vartheta \cdot\diff \phi_{\hat{X}_j \check{X}^{-1}} \cdot\diff \phi^{\xi_0} \cdot \mmap, \label{eq:J_matrix} \\
        C_j &= \diff h \cdot \diff \phi_{\hat{X}_j} \cdot \diff \phi^{\xi_0} \cdot \mmap, \label{eq:C_matrix}
    \end{align}
    and $\alpha_j > 0$ is a step size that may be determined by a chosen line search method.
\end{theorem}

\begin{proof}
First note that any perturbation $\Delta \hat{X}_j$ with $\Delta \in \stab_\phi(\xi_0)$ leaves the residual $\tilde{r}$ unchanged.
Therefore, it suffices to consider perturbations of the form $\Delta = \mmap \delta$, where $\delta \in \R^m$.
Fix $j \in \N$, then the residual for a perturbation $\Delta_j = \mmap \delta_j\in \mathfrak{g}$ about the current value $\hat{X}_j$ is given by
\begin{align*}
    \tilde{r}(&\exp(\mmap \delta_j)\hat{X}_j) \\ 
    &= \begin{pmatrix}
        \vartheta(\phi(\exp(\Delta_j)\hat{X}_j \check{X}^{-1}, \xi_0)) \\
        y - h(\phi(\exp(\Delta_j)\hat{X}_j,\xi_0))
    \end{pmatrix} \\
    &= \begin{pmatrix}
        \vartheta(\phi(\hat{X}_j \check{X}^{-1}, \phi(\exp(\Delta_j), \xi_0))) \\
        y - h(\phi_{\hat{X}_j}(\phi_{\xi_0}(\exp(\Delta_j))))
    \end{pmatrix} \\
    &= \begin{pmatrix} 
        \vartheta \circ \phi_{\hat{X}_j \check{X}^{-1}} \circ \phi^{\xi_0} \circ \exp(\Delta_j) \\
        y - h \circ \phi_{\hat{X}_j} \circ \phi_{\xi_0} (\exp(\Delta_j))
    \end{pmatrix} \\
    &= \tilde{r}(\hat{X}_j) +
    \begin{pmatrix} 
        \diff \vartheta \cdot\diff \phi_{\hat{X}_j \check{X}^{-1}} \cdot\diff \phi^{\xi_0} \\
        - \diff h \cdot \diff \phi_{\hat{X}_j} \cdot \diff \phi^{\xi_0}
    \end{pmatrix} \Delta_j + O(\vert \Delta_j \vert^2) \\
    &= \tilde{r}(\hat{X}_j) +
    \begin{pmatrix} 
        J_j \\
        -C_j
    \end{pmatrix} \delta_j + O(\vert \delta_j \vert^2).
\end{align*}
The Gauss-Newton update is therefore given by
\begin{align*}
    W \delta_j
    &= - \begin{pmatrix} J_j \\ -C_j \end{pmatrix}^\top \diag(\check{\Sigma}^{-1}, R^{-1}) \tilde{r}(\hat{X}_j) \\
    &= - \begin{pmatrix} J_j^\top \check{\Sigma}^{-1} &
         -C_j^\top R^{-1} \end{pmatrix} \begin{pmatrix}
        \vartheta(\phi(\hat{X}_j \check{X}^{-1}, \xi_0)) \\
        y - h(\phi(\hat{X}_j,\xi_0))
    \end{pmatrix} \\
    &= - J_j^\top \check{\Sigma}^{-1} \vartheta(\phi(\hat{X}_j \check{X}^{-1}, \xi_0)) 
        \\&\hspace{0.5cm}
        + C_j^\top R^{-1} (y - h(\phi(\hat{X}_j,\xi_0))) \\
    &= - J_j^\top \check{\Sigma}^{-1} \varepsilon_j 
        + C_j^\top R^{-1} (y - h(\phi(\hat{X}_j,\xi_0))),
\end{align*}
where $W \in \Sym_+(m)$ is the approximate Hessian along the image of $\Theta$ in $\gothg$, given by
\begin{align*}
    W 
    &:= \begin{pmatrix} J_j \\ -C_j \end{pmatrix}^\top \diag(\check{\Sigma}^{-1}, R^{-1}) \begin{pmatrix} J_j \\ -C_j \end{pmatrix} \\
    &= J_j^\top \check{\Sigma}^{-1} J_j + C_j^\top R^{-1} C_j \\
    &= \check{\Sigma}_j^{-1} + C_j^\top R^{-1} C_j ,
\end{align*}
where $\check{\Sigma}_j$ is defined in \eqref{eq:check_Sigma_dfn}. 
Using the matrix inversion lemma \cite[Corollary 2.8.8]{bernstein2009matrix},
\begin{align*}
    W^{-1} C_j^\top
    &= (\check{\Sigma}_j - \check{\Sigma}_j C_j^\top (R + C_j \check{\Sigma}_j C_j^\top)^{-1} C_j \check{\Sigma}_j)C_j^\top \\
    &= \check{\Sigma}_j C_j^\top - \check{\Sigma}_j C_j^\top (R + C_j \check{\Sigma}_j C_j^\top)^{-1} C_j \check{\Sigma}_j C_j^\top \\
    &= \check{\Sigma}_j C_j^\top + \check{\Sigma}_j C_j^\top (R + C_j \check{\Sigma}_j C_j^\top)^{-1} R
    \\&\hspace{0.5cm}
    - \check{\Sigma}_j C_j^\top (R + C_j \check{\Sigma}_j C_j^\top)^{-1} (R + C_j \check{\Sigma}_j C_j^\top) \\
    &= \check{\Sigma}_j C_j^\top (R + C_j \check{\Sigma}_j C_j^\top)^{-1} R \\
    &= K_j R.
\end{align*}
Therefore, the update $\delta_j$ is computed as
\begin{align*}
    \delta_j
    &= W^{-1} (- J_j^\top \check{\Sigma}^{-1} \varepsilon_j + C_j^\top R^{-1} (y - h(\phi(\hat{X}_j,\xi_0)))) \\
    &= - W^{-1} (\check{\Sigma}_j^{-1} J_j^{-1} \varepsilon_j - C_j^\top R^{-1} (y - h(\phi(\hat{X}_j,\xi_0)))) \\
    &= - W^{-1} \Big((\check{\Sigma}_j^{-1} + C_j^\top R^{-1} C_j) J_j^{-1} \varepsilon_j 
    \\ &\hspace{1.cm}
    - C_j^\top R^{-1} (y - h(\phi(\hat{X}_j,\xi_0)) + C_j J_j^{-1} \varepsilon_j ) \Big) \\
    &= -J_j^{-1} \varepsilon_j 
    + W^{-1} C_j^\top R^{-1} (y - h(\phi(\hat{X}_j,\xi_0)) + C_j J_j^{-1} \varepsilon_j ) \\
    &= -J_j^{-1} \varepsilon_j 
    + K_j (y - h(\phi(\hat{X}_j,\xi_0)) + C_j J_j^{-1} \varepsilon_j ),
\end{align*}
This verifies the formula for the update $\Delta_j = \mmap \delta_j$, so all that remains is to verify the updated covariance $\hat{\Sigma}_{j+1}$.
The covariance $\hat{\Sigma}_{j+1}$ is simply found as the inverse of the approximate Hessian of the cost; that is,
\begin{align*}
    \hat{\Sigma}_{j+1} &= (\check{\Sigma}_j^{-1} + C_j^\top R^{-1} C_j)^{-1}
    = (I - K_j C_j) \check{\Sigma}_j,
\end{align*}
where the last equality follows once more from the matrix inversion lemma \cite[Corollary 2.8.8]{bernstein2009matrix}.
\end{proof}

The proposed iteration addresses a cost $\tilde{S}$ that has been lifted to the Lie group from a cost $S$ that was posed on a (possibly lower-dimensional) homogeneous space, in contrast to existing works \cite{paper_Bourmaud_iteritive_LG_IKF,Lu_etal_2025,Goffin_etal_2026} that address costs directly posed on the Lie group. 
The resulting procedure appears similar to prior solutions on the surface, but special care is taken to ensure that the perturbation on the Lie group at each step is a descent direction, and that the degenerate directions in the cost function (associated with lifting) do not cause instability.
Additionally, the final covariance (inverse Hessian) $\hat{\Sigma}_j$ is deliberately computed in the local coordinates of the manifold, rather than in the Lie algebra of the group, to ensure that the final approximation of the cost function is also well-defined on the homogeneous space.

In \eqref{eq:check_Sigma_dfn}, the covariance of the prior information is explicitly modified by the nonlinear change of coordinates associated with the changing posterior $\hat{X}_j$.
This corresponds exactly to the reset step that is often applied to improve performance of geometric filters, and is likewise associated with nonlinear modification of filter error coordinates \cite{paper_Markley_JGCD_2003,paper_Mueller_etal_JGCD_2017,paper_Ge_etal_CDC_2022,Paper_Ge_van_Goor_Mahony_LCSS_2024,ge_etal_CEP_2026}.
The appearance of the reset step as a byproduct of the Gauss-Newton formulation is similarly recovered by \cite{paper_Bourmaud_iteritive_LG_IKF,paper_Liu_Chen_Zhang_CDC_2023}, although these prior works only addressed systems with Lie-group state spaces.
In contrast, the iterative IEKF by Goffin \textit{et al.} \cite{Goffin_etal_2026} deliberately omits the reset step to preserve compatibility properties of the IEKF for systems with unobservability.
The present work is the first to derive the IterEqF for systems on homogeneous spaces (not just Lie groups) and confirms the reset step as a natural consequence of the Gauss-Newton formulation of the update step.

\section{Example: Range-based Localisation}
\label{sec:example}

To demonstrate the IterEqF, we conducted simulations of a kinematic robot equipped with a Ultra-Wide Band (UWB) range sensor, providing distances of the robot to beacons at known locations.
Formally, the state of the robot is $(\theta, p) \in \SO(2) \times \R^2$, where $\theta$ denotes the angle of the robot and $p$ denotes the position of the robot, both with respect to some chosen stationary reference frame.
For simplicity, the state space $\SO(2) \times \R^2$ is identified with the Lie group $\SE(2)$, and the robot's pose is written as $P = (R, p) \in \SE(2)$, where $R \in \SO(2)$ is the rotation matrix associated with the angle $\theta$.
The kinematics of the robot are then given by
\begin{align*}
    \dot{P} = f_{(\omega, v)}(P) = P U, &&
    U = \begin{pmatrix}
        \omega^\times & v \eb_1 \\ 0_{1\times 2} & 0
    \end{pmatrix} \in \se(2),
\end{align*}
where $\omega \in \R$ is the angular velocity of the robot, and $v \in \R$ is the forward translational velocity of the robot, both expressed in the body-fixed frame, and $\eb_1 \in \R^2$ is the first basis vector.
Given a UWB beacon located at $\ell_i$ in the reference frame, the measured range from the robot to the beacon is given by
\begin{align*}
    y_i = h_i(P) &= \| p - \ell_i \|.
\end{align*}
Here, $y_i \in \R$ is the measured range and $h_i : \SE(2) \to \R$ is the measurement function corresponding to beacon $i$.

\subsection{EqF Construction}

The first step in developing an EqF for the problem is to select a transitive group action on the state space.
We choose to use $\phi : \SE(2) \times \SE(2) \to \SE(2)$ given by
\begin{align*}
    \phi(X, P) = P X.
\end{align*}
It is easily verified that this is indeed a transitive group action.
Consider the map $\Lambda : \SE(2) \times \R^2 \to \se(2)$ given by 
\begin{align*}
    \Lambda(P, (\omega, v)) := \begin{pmatrix} \omega^\times & v \eb_1 \\ 0_{1\times 2} & 0 \end{pmatrix}.
\end{align*}
Then $\Lambda$ is clearly a lift as it satisfies \eqref{eq:lift_condition}, and, in particular, it is the only lift for the system considered since $\phi$ is free as well as transitive.
Notably, $\Lambda$ is independent of its first argument, which greatly simplifies the computation of the EqF matrices $A_t$ and $B_t$.
Next, we choose the normal coordinates $\vartheta : \SE(2) \to \R^3$ to be given by
\begin{align*}
    \vartheta(P) &:= \log(P)^\vee, &
    \vartheta^{-1}(\varepsilon) &:= \exp(\varepsilon^\wedge),
\end{align*}
where the wedge map $\cdot^\wedge : \R^3 \to \se(2)$ is defined by
\begin{align*}
    \begin{pmatrix}
        \omega \\ u_1 \\ u_2
    \end{pmatrix}^\wedge
    := \begin{pmatrix}
        0 & - \omega & u_1 \\
        \omega & 0 & u_2 \\
        0 & 0 & 0
    \end{pmatrix},
\end{align*}
and the vee map $\cdot^\vee : \se(2) \to \R^3$ is its inverse.
The origin $\xi_0 = \vartheta^{-1}(0)$ is simply the identity in $\SE(2)$.

Implementation of the EqF requires computing formulas for the matrices $A_t, B_t, C_j, J^\vartheta(\varepsilon)$.
For the matrix $A_t$ given by \eqref{eq:A_matrix}, the fact that $\Lambda$ is independent of the state means that $\diff \Lambda^u \equiv 0$, and thus $A_t \equiv 0$ also.
The matrix $B_t$ is given as
\begin{align*}
    B_t = \begin{pmatrix}
        1 & 0_{1\times 2} \\
        -1^\times \hat{p} & \hat{R}
    \end{pmatrix} \begin{pmatrix}
        1 & 0 \\ 0_{2\times 1} & \eb_1
    \end{pmatrix}
    = \begin{pmatrix}
        1 & 0 \\
        -1^\times \hat{p} & \hat{R} \eb_1
    \end{pmatrix},
\end{align*}
where $\hat{R} \in \SO(2)$ and $\hat{p} \in \R^2$ are the rotation and translation components of $\hat{X}$, respectively, and 
\begin{align*}
    q^\times = \begin{pmatrix}
        0 & -q \\ q & 0
    \end{pmatrix},
\end{align*}
for all $q \in \R$.
For a given $\hat{X}_j = (\hat{R}_j, \hat{p}_j) \in \SE(2)$, the matrix $C_j$ corresponding to beacon $i$ is given by
\begin{align*}
    C_j = \frac{(\hat{p}_j - \ell_i)^\top}{\| \hat{p}_j - \ell_i \|}\begin{pmatrix} 1^\times \hat{p}_j & I_2 \end{pmatrix}.
\end{align*}
For multiple beacons, the individual $C_j$ matrices may be vertically stacked.
Finally, the matrix $J^\vartheta$ is given by
\begin{align*}
    J^\vartheta(\varepsilon)[\delta]
    &= \diff \log (\exp(\varepsilon)) [\delta \exp(\varepsilon)] \\
    &= I_3 - \frac{1}{2} \ad_\varepsilon + \beta(\varepsilon_\theta) \ad_\varepsilon^2, \\
    \beta(\theta) &:= \frac{1}{\theta^2} - \frac{1 + \cos(\theta)}{2\theta \sin(\theta)}.
\end{align*}
This is the inverse of what is referred to as the left Jacobian of $\SE(2)$ \cite{paper_Sola_eta_al_micro}.

\subsection{Simulation Results}

\begin{figure}[htb]
    \centering
    \includegraphics[width=0.5\linewidth]{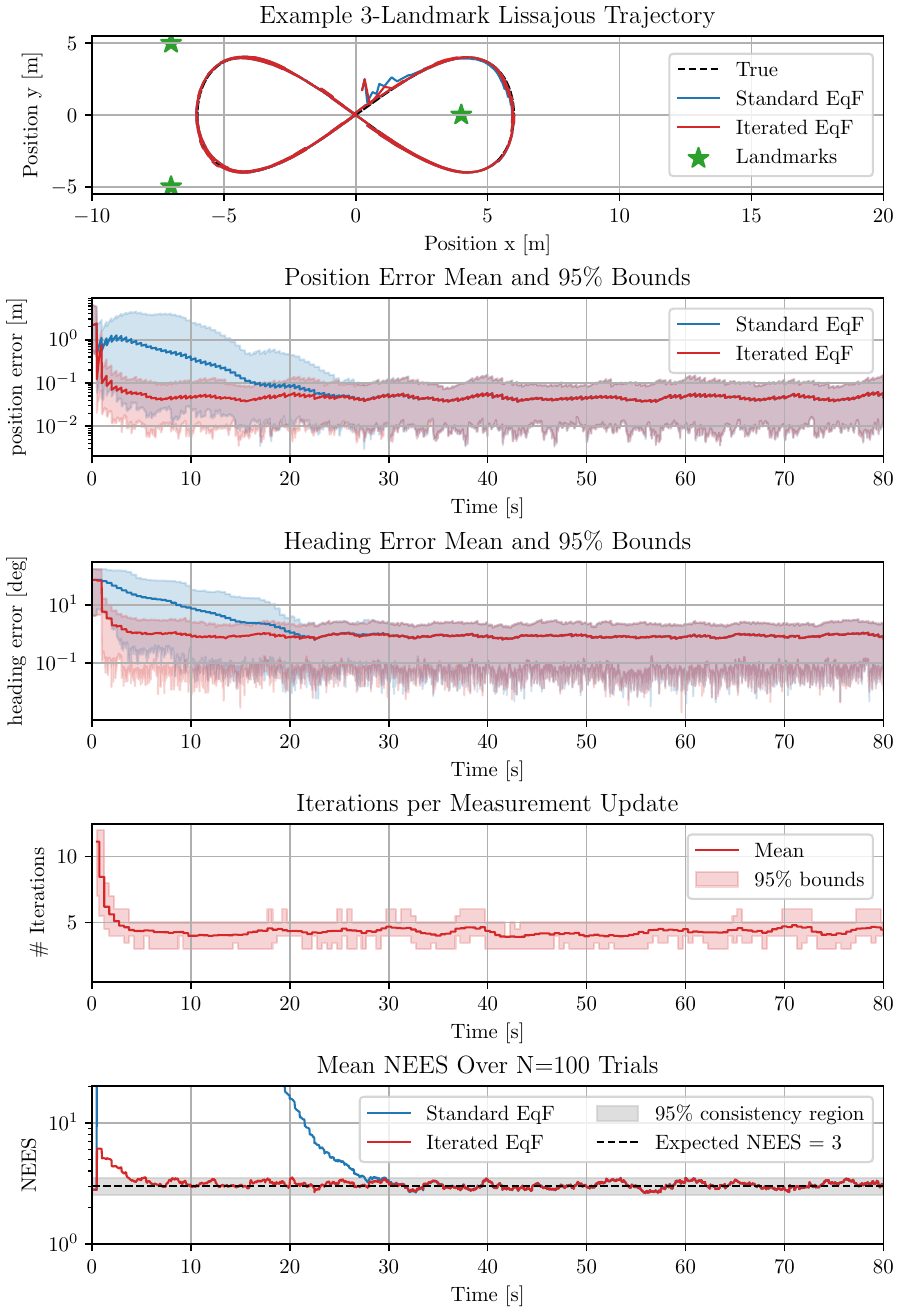}
    \caption{The results of a simulation comparing the IterEqF with a standard (single update) EqF.
    The first subplot shows the evolutions of example trajectories compared to the ground truth.
    The second and third show the evolution of position errors and heading errors over 100 trials.
    The fourth subplot shows the distribution of the number of iterations used by the IterEqF at each time.
    The last subplot shows the average NEES of the filters over time.}
    \label{fig:3landmark}
\end{figure}

To compare the IterEqF with the standard EqF, we simulated a mobile robot travelling along a figure-eight Lissajous curve with its position determined by $p(t) = (6 \sin(\tfrac{2\pi t}{40}), 4 \sin(\tfrac{4\pi t}{40}))$ and its angle $\theta(t)$ determined such that the robot's heading remains tangent to the path.
The period of the Lissajous is 40~s, and the simulation was run for 80~s using 50~Hz Euler integration for continuous-time processes.
The forward and angular velocities were corrupted with zero-mean Gaussian noise of standard deviations 0.01~m/s and 0.01~rad/s, respectively.
We simulated 3 landmarks located at 
\begin{align*}
    \ell_1 = (4,0), && \ell_2 = (-7,5), && \ell_3 = (-7,-5).
\end{align*}
The range measurements were provided at 2~Hz and were corrupted with zero-mean Gaussian noise of standard deviation 0.1~m.
The initial filter covariance was set to $\Sigma(0) = \diag((\tfrac{\pi}{2})^2, 2.0^2, 2.0^2)$, and the initial pose estimate was given by $\hat{X}(0) = \exp(\varepsilon_0^\wedge) P(0)$, where $P(0) \in \SE(2)$ is the true initial pose and $\varepsilon_0 \sim N(0, \Sigma(0))$.
The simulation was implemented in Python, and the code is publicly available\footnote{\href{https://github.com/pvangoor/iterative_equivariant_filter}{\texttt{github.com/pvangoor/iterative\_equivariant\_filter}}}.

Figure \ref{fig:3landmark} shows the result of the simulation repeated over 100 trials.
The top subplot shows example trajectories of the iterative and standard EqF compared to the ground truth.
The heading angle of the standard EqF converges slowly, leading to increasing position errors between measurements, while the IterEqF heading angle converges quickly.
The second and third subplot show that the mean position and heading angle both converge significantly more quickly for the IterEqF as compared to the standard EqF, which corresponds to the high number of iterations (shown in the fourth subplot) during the initial 5~s of the simulation.
Around 30~s, both the iterative and standard EqF have completely converged, and their accuracy is matched from this point on.
The fluctuations in the estimation errors are associated with regions where the position and heading are less easily inferred from the range measurements.
The bottom subplot shows the filters' Normalised Estimation Error Squared (NEES) $\varepsilon^\top \Sigma^{-1} \varepsilon$ over time, which demonstrates that the IterEqF is significantly more statistically consistent than the standard EqF during the transient phase.
In summary, while the IterEqF necessarily incurs a greater computational cost, it is able to converge far more rapidly than a standard EqF thanks to these iterations, which is especially relevant to problems like range-based localisation where observability may be challenging.

\section{Conclusion}
\label{sec:conclusion}

This paper presents the IterEqF.
The motivation behind deriving the IterEqF is to exploit the features of both iteration and equivariance thereby leading to a high-performance filter well suited to navigation problems.
The IterEqF correction step is found by solving a weighted nonlinear least squares problem using the Gauss-Newton algorithm.
Interestingly, the reset step naturally appears in the iterative correction, circumventing any justification for why the reset step is or is not needed. 
Monte-Carlo simulation results verify that the proposed IterEqF realizes superior convergence properties, in terms of both estimation error and covariance, relative to the standard EqF. 

\printbibliography

\end{document}

%% file: preamble.tex
\usepackage{graphicx}
\usepackage{amsmath,amssymb,amsfonts}

\usepackage{amsthm}
\usepackage{algorithm,algorithmic}
\usepackage[hidelinks=true]{hyperref}

\usepackage{xcolor}

\DeclareMathOperator{\im}{im}

\DeclareMathOperator{\Ad}{Ad}
\DeclareMathOperator{\ad}{ad}

\newcommand{\diff}[0]{{\mathrm{d}}}
\newcommand{\ddt}[0]{\tfrac{\diff}{\diff t}}

\newcommand{\T}{\mathrm{T}}

\newcommand{\calM}{\mathcal{M}}
\newcommand{\calN}{\mathcal{N}}
\newcommand{\calU}{\mathcal{U}}
\newcommand{\R}{\mathbb{R}}
\newcommand{\N}{\mathbb{N}}
\newcommand{\grpG}{\mathbf{G}}
\newcommand{\gothg}{\mathfrak{g}}

\newcommand{\stab}{\mathfrak{stab}}
\newcommand{\Sym}{\mathrm{Sym}}
\newcommand{\diag}{\mathrm{diag}}
\newcommand{\eb}{\mathbf{e}}

\newcommand{\SO}{\mathbf{SO}}
\newcommand{\SE}{\mathbf{SE}}

\newcommand{\se}{\mathfrak{se}}